\documentclass[11pt]{article}

\usepackage{fullpage}
\usepackage{latexsym}
\usepackage{amssymb,amsfonts,amsmath,amsthm}
\usepackage{graphicx}
\usepackage{epsfig}
\usepackage{pstricks}
\usepackage[all]{xy}

\def\01{\{0,1\}}
\newcommand{\ceil}[1]{\lceil{#1}\rceil}

\newcommand{\eps}{\varepsilon}

\newcommand{\Exp}{{\mathbb{E}}}

\newcommand{\inpc}[2]{\langle{#1},{#2}\rangle} 
\newcommand{\Tr}{\mbox{\rm Tr}}

\renewcommand{\Pr}{\mbox{\rm Pr}}

\newcommand{\half}{{\frac{1}{2}}}
\newcommand{\set}[1]{{\left\{#1\right\}}}

\newcommand{\pmset}[1]{\{\pm1\}^{#1}}
\newcommand{\polylog}{\mbox{\rm polylog}}
\newcommand{\probdist}[1]{\mathcal{P}(#1)}
\newcommand{\densop}[1]{\mathcal{B}_+^1(\mathcal{H}_{#1})}

\newcommand{\beq}{\begin{equation}}
\newcommand{\eeq}{\end{equation}}
\newcommand{\beqn}{\begin{equation*}}
\newcommand{\eeqn}{\end{equation*}}
\newcommand{\beqr}{\begin{eqnarray}}
\newcommand{\eeqr}{\end{eqnarray}}
\newcommand{\beqrn}{\begin{eqnarray*}}
\newcommand{\eeqrn}{\end{eqnarray*}}

\newtheorem{definition}{Definition}
\newtheorem{theorem}{Theorem}
\newtheorem{lemma}[theorem]{Lemma}
\newtheorem{corollary}[theorem]{Corollary}

\renewcommand{\qed}{\hfill{\rule{2mm}{2mm}}}
\renewenvironment{proof}[1][]{\begin{trivlist}
\item[\hspace{\labelsep}{\bf\noindent Proof#1:\/}] }{\qed\end{trivlist}}

\begin{document}

\title{Locally Decodable Quantum Codes}
\author{
Jop Bri\"{e}t\thanks{jop.briet@cwi.nl. Partially supported by a Vici grant from the Netherlands Organization for Scientific Research (NWO), and by the European Commission under the Integrated Project Qubit Applications (QAP) funded by the IST directorate
as Contract Number 015848.}\\CWI
\and
Ronald de Wolf\thanks{rdewolf@cwi.nl. Partially supported by a Veni grant from the Netherlands Organization for Scientific Research (NWO), and by the European Commission under the Integrated Project Qubit Applications (QAP) funded by the IST directorate
as Contract Number 015848.}\\CWI}
\date{}
\maketitle

\begin{abstract}
We study a quantum analogue of locally decodable error-correcting codes.
A $q$-query \emph{locally decodable quantum code} encodes $n$ classical bits in an $m$-qubit state,
in such a way that each of the encoded bits can be recovered with high probability by a measurement
on at most $q$ qubits of the quantum code, even if a constant fraction of its qubits have been corrupted adversarially.
We show that such a quantum code can be transformed into a \emph{classical} $q$-query locally decodable code of
the same length that can be decoded well on average (albeit with smaller success probability and noise-tolerance).
This shows, roughly speaking, that $q$-query quantum codes are not significantly better than $q$-query classical codes, at least for constant or small $q$.
\end{abstract}

\section{Introduction}

\subsection{Setting}

Locally decodable codes (LDCs) have received much attention in the last decade.
They are error-correcting codes that encode $n$ bits into $m$ bits, with the usual
error-correcting properties, and the additional feature that any one of
the $n$ encoded bits can be recovered (with high probability) by a randomized
decoder that queries at most $q$ bits in the codeword, for some small $q$.
In other words, to decode small parts of the encoded data, we only need to look
at a small part of the codeword instead of ``unpacking'' the whole thing.
Precise definitions will be given in the next sections.
Such codes are potentially useful in their own right (think of decoding small pieces
from a large encoded library), and also have a variety of applications in complexity
theory and cryptography.  For instance, it is well known that they can be turned into
private information retrieval schemes and vice versa (see the appendix).
For further details about such connections, we refer to Trevisan's survey~\cite{trevisan:eccsurvey} and the references therein.

The most interesting question about LDCs is the tradeoff between
their length $m$ and the number of queries $q$.  The former measures
the space efficiency of the code, while the latter measures the efficiency of decoding.
The larger we make $q$, the smaller we can make $m$.
On one extreme, if we allow $q=\polylog(n)$ queries, the codelength $m$
can be made polynomial in~$n$~\cite{bfls:checking}.
On the other extreme, for $q=1$ and sufficiently large $n$,
LDCs do not exist at all~\cite{katz&trevisan:ldc}. For $q=2$ they do exist
but need exponential length, $m=2^{\Theta(n)}$~\cite{kerenidis&wolf:qldcj}.
Between these two extremes, interesting but hard questions persist.
In particular, we know little about the length of LDCs with constant $q>2$.
The best upper bounds for $q=3$ are Yekhanin's recent construction~\cite{yekhanin:3ldc}:
he gives 3-query LDCs with length $m=2^{O(n^{1/t})}$ for every \emph{Mersenne prime} $p=2^t-1$.
Currently only finitely many Mersenne primes are known (the largest has $t=32582657$),
but it has been conjectured that there are infinitely many.
For $q>3$, shorter codes may be derived by combining Yekhanin's codes with
the recursive constructions of Beimel et al.~\cite{bikr:improvedpir}.
However, all these bounds still have superpolynomial length $m$ for every constant $q$.
On the lower bound side, the best we know for $q>2$ is $m=\Omega\left((n/\log n)^{1+1/(\ceil{q/2}-1)}\right)$
\cite{katz&trevisan:ldc,kerenidis&wolf:qldcj,woodruff:ldclower} (these bounds are stated for fixed 
success probability and noise rate). For instance for $q=3$ and $q=4$, our best lower bounds are slightly
less than $n^2$.

Interestingly, the best known lower bounds were obtained using tools from quantum information theory.
It is thus a natural question to consider also the potential \emph{positive} effects of quantum:
can we construct much shorter $q$-query locally decodable codes by somehow harnessing the 
power of quantum states and quantum algorithms?
There are two natural ways to generalize classical locally decodable codes to the quantum world:
\begin{itemize}
\item We can keep the code classical, but allow $q$ quantum queries.
This means we can query positions of the codeword in quantum superposition,
and process the results using quantum circuits.
This approach was investigated in~\cite{kerenidis&wolf:qldcj}.
A $q$-query quantum decoder can simulate a $2q$-query classical decoder with high success probability,
and this simulation can be made exact if the classical decoder took the parity of its $2q$ bits.
This implies for instance that Yekhanin's 3-query LDC can be decoded
by only 2 quantum queries.  In contrast, we know that every 2-query LDC needs
length $2^{\Theta(n)}$. Allowing quantum queries thus results in very large savings
in $m$ when we consider a fixed number of queries $q$.
\item We can also make the code itself quantum: instead of encoding an $n$-bit $x$
into an $m$-bit string $C(x)$, we could encode it into an $m$-qubit state $Q(x)$.
A $q$-query decoder for such a code would select up to $q$ qubits of the state
$Q(x)$, and make some 2-outcome measurement on those qubits to determine its output.
In this case our notion of noise also needs to be generalized: instead of up to $\delta m$ bitflip-errors,
we now allow any set of up to $\delta m$ qubits of the $m$-qubit state $Q(x)$ to be arbitrarily changed.%
\footnote{While a classical LDC can be reused as often as we want, a quantum code has the problem that a measurement
made to predict one bit changes the state, so predicting another bit based on the changed state
may give the wrong results. However, if the error probability is small then the changes incurred
by each measurement will be small as well, and we can reuse the code many times with reasonable confidence.

Another issue is that more general decoders could be allowed.  For instance, we could consider
allowing any quantum measurement on the $m$-qubit state that can be written as a linear combination of
$m$-qubit Pauli-matrices that have support on at most $q$ positions.  This is potentially stronger
than what we do now (it is an interesting open question whether it is really stronger).
However, we feel this is a somewhat unnatural formalization
of the idea that a measurement should be localized to at most $q$ qubits.
Our current set-up, where we classically select up to $q$ positions and then apply an arbitrary
quantum measurement to those $q$ qubits, seems more natural.}
\end{itemize}

\subsection{Our results}

In this paper we investigate the second kind of code,
which we call a ``$q$-query locally decodable quantum code'', or $q$-query LDQC.
The question is whether the ability to encode our $n$ bits into a \emph{quantum} state enables us
to make codes much shorter.  There are some small examples where quantum encodings achieve things that are impossible for classical encodings.
For example, Ambainis et al.~\cite{antv:densej} give an example of an encoding of 2 classical bits
into 1 qubit, such that each of the bits---though not both simultaneously---can be recovered
from the qubit with success probability~0.85. They even cite an example due to Chuang where 3 bits are encoded
into 1 qubit, and each bit can be recovered with success probability~0.78.
However, they also show that asymptotically large savings are not possible in their setting of random access codes
(explained in Section~\ref{ssecqrac} below). 

Their setting, however, considers neither noise nor local decodability, and hence
does not answer our question about locally decodable codes:
can LDCs be made significantly shorter if we allow quantum instead of classical encodings?
Our main result in this paper is a negative answer to this question:
essentially it says that $q$-query locally decodable quantum codes can be turned
into $q$-query locally decodable classical codes of the same length, with some deterioration in their other parameters.
The precise statement of this result (Corollaries~\ref{ldqc=rldc:cor} and~\ref{ldqc=ldc:cor}) is a little bit dirty.
We obtain a cleaner statement for so-called ``smooth (quantum) codes'', which have the property
that they query the codewords fairly uniformly.
These smooth (quantum) codes can be converted into LD(Q)Cs and vice versa.
For these, the precise statement is as follows (Theorem~\ref{sqc=rsc:thm}).

Suppose we are given a smooth quantum code of $m$ qubits from which we can recover (with success probability at least $1/2+\eps$) 
each bit $x_i$ of the encoded $n$-bit string $x$, while only looking at $q$ qubits of the state. Let $\mu$ be a distribution on the $n$-bit inputs.
Then we can construct a \emph{randomized} classical code $R$ of the same length
(for each $x$, the ``codeword'' is a \emph{distribution} over $m$-bit strings)
from which we can recover each $x_i$ with $\mu$-average success probability at least $1/2+\eps/4^{q+1}$, while
only looking at $q$ bits of the codeword. Thus a $q$-query quantum code is turned into a $q$-query classical code of the same length.
For those who do not like the idea of encoding $x$ into a \emph{distribution} $R(x)$,
we can turn the randomized code $R$ into a \emph{deterministic} code $C$,
where $C(x)$ is a fixed $m$-bit codeword instead of a distribution, at the expense
of correctly decoding only a constant fraction of all indices $i$ instead of all $n$ of them (Corollary~\ref{sqc=sc:cor}).

Since all known lower bounds on LDCs also apply to randomized classical codes that work well under 
a uniform distribution $\mu$ on the $n$-bit strings, those lower bounds immediately carry over to LDQCs.  
In particular we obtain as corollaries of our result:
\begin{itemize}
\item For sufficiently large $n$, 1-query LDQCs do not exist for any length $m$ (from \cite{katz&trevisan:ldc}).%
\footnote{Actually, this result can more easily be shown directly, by combining Katz and Trevisan's proof for classical codes
with the quantum random access code lower bound mentioned below in Section~\ref{ssecqrac}.}
\item 2-query LDQCs need length $m=2^{\Theta(n)}$ (from \cite{kerenidis&wolf:qldcj}).
\item For every constant $q$, $q$-query LDQCs need length $m=\Omega\left((n/\log n)^{1+1/(\ceil{q/2}-1)}\right)$ (from~\cite{kerenidis&wolf:qldcj}).
\end{itemize}
\paragraph{Techniques.}
Our main technique is to apply to the $m$ qubits of the quantum code a randomly selected sequence of $m$ Pauli measurements.
The randomized ``codeword'' $R(x)$ will be the probability distribution on $m$-bit outcomes 
that results from applying such a measurement to the quantum state $Q(x)$.
The main part of our proof is to show that there exists a choice of Pauli measurements that roughly preserves correct 
decodability for all indices $i$.

\section{Preliminaries}

We write $[n]$ for the set $\set{1,\dots,n}$. 
We use $\probdist{S}$ to denote the set of all probability distributions (or random variables) on set $S$.
If $z$ is distributed according to the distribution of a random variable $Z$, we write $z\sim Z$. 
We will use this when taking probabilities $\Pr_{z\sim Z}$ or expectations $\Exp_{z\sim Z}$.
Probabilities and expectations with a subscript `$i\in S$' should be read as taken over a uniformly random $i\in S$. 
Below we give a brief overview of the concepts of quantum mechanics used here, see~\cite{nielsen&chuang:qc,preskill:notes} for more extensive introductions.

\paragraph{Quantum states.}
In quantum mechanics, a physical system is mathematically represented by a complex Hilbert space. A $d$-dimensional complex Hilbert space consists of all $d$-dimensional vectors with complex entries, endowed with the standard inner product. The \emph{state} of a physical system is in turn represented by a density operator (a positive semidefinite linear operator with trace~1) acting on a Hilbert space. 
We use $\densop{d}$ to denote the set of all density operators on a $d$-dimensional complex Hilbert space.
States in two-dimensional Hilbert spaces are called \emph{qubits}. Density operators of rank~1 are called \emph{pure states}.

\paragraph{Measurements.}
Information about the state of a physical system can only be obtained by doing a \emph{measurement}. The most general $k$-outcome measurement can be defined as a set $\set{A_1,\dots,A_k}$ of $k$ positive semidefinite matrices that satisfy $\sum_{i=1}^k A_i=I$. The probability that the measurement of a system in a state $\rho$ yields the $i$'th outcome  is $\Tr(A_i\rho)$. Hence, the measurement yields a random variable $A(\rho)$ with ${\Pr[A(\rho)=i] = \Tr(A_i\rho)}$.
With a measurement that has outcomes $+1$ and $-1$ (and corresponding operators $A^+$ and $A^-$) we associate an operator $A = A^+ - A^-$. 
The expected value of this measurement on a state $\rho$ is then $\Tr(A\rho)$. Note that this equals the difference between
the probabilities of outcomes $+1$ and $-1$, respectively.

\paragraph{Pauli matrices.}
The one-qubit Pauli operators are given by 
$$
I = \left(\begin{array}{rr}1 & 0\\ 0&1\end{array}\right), \
X = \left(\begin{array}{rr}0&1\\ 1&0\end{array}\right), \
Y = \left(\begin{array}{rr}0&-i\\i&0\end{array}\right), \mbox{ and }
Z = \left(\begin{array}{rr}1&0\\ 0&-1\end{array}\right).
$$ 
For integer $k\geq 1$, the set of $k$-qubit Pauli operators is $\mathcal{P}_k:=\set{I,X,Y,Z}^{\otimes k}$. 
These $4^k$ matrices form an orthonormal basis for the space of all $2^k\times 2^k$ complex matrices endowed with the
inner product $\inpc{A}{B}=\frac{1}{2^k}\Tr(A^\dagger B)$.
Each Pauli operator $S\in\mathcal{P}_k$ has a unique decomposition $S = S^+ - S^-$, with $S^+$ and $S^-$ orthogonal projectors that satisfy $S^+ + S^- = I$. For this reason we associate a unique two-outcome measurement $\{S^+,S^-\}$ with each such $S$. A Pauli measurement $S\in\mathcal{P}_k$ of a $k$-qubit state $\rho$ yields a $\pm 1$-valued random variable $S(\rho)$ with expected value $\Tr(S\rho)$.
However, we can also view $S\in\mathcal{P}_k$ as $k$ separate one-qubit Pauli measurements, to be applied to the $k$ qubits of the state, respectively.
When viewed in this way, the result of measuring $\rho$ is an $k$-bit random variable, i.e., a probability distribution on $\pmset{k}$.
The product of those $k$ bits equals the $\pm 1$-valued random variable $S(\rho)$ mentioned before.

\paragraph{Super-operators.}
A super-operator is a mathematical representation of the most general transformation of a quantum state allowed by the laws of quantum mechanics. A super-operator $\mathcal{E}$ can be defined by a finite set $\set{E_1,\dots,E_k}$ of linear operators (known as Kraus operators) that satisfy $\sum_{i=1}^k E_i^{\dagger}E_i=I$. The corresponding operation on a state $\rho$  yields another density operator, $\mathcal{E}(\rho)=\sum_{i=1}^k E_i\rho E_i^{\dagger}$.  This $\mathcal{E}(\rho)$ may act on a Hilbert space of a possibly different dimension, though we will not need that here. We say that $\mathcal{E}$ ``acts trivially'' on, say, the first
qubit of the state if all its Kraus operators have the form $E_i=I\otimes E'_i$ for some $E'_i$ acting on all but the first qubit.

\section{Codes}

It will be convenient to write bits as $\pm1$ instead of 0/1.
With this convention, if random variable $A\in\{\pm 1\}$ predicts bit $x_i\in\{\pm 1\}$, we can write the bias of this prediction as an expectation: $\Exp[A\cdot x_i]=\Pr[A=x_i]-\Pr[A\neq x_i]$.  Note that $\Pr[A=x_i]\geq 1/2+\eps$ iff $\Exp[A\cdot x_i]\geq 2\eps$.

\subsection{Classical codes}\label{clascodes:section}

We start with classical codes.
The formal definition of a locally decodable code is as follows.
It involves a decoder $\mathcal{A}$ that receives input $i\in[n]$ and oracle access to a string $y\in\pmset{m}$,
usually written as a superscript to $\mathcal{A}$.
This $y$ will be a codeword $C(x)\in\pmset{m}$ corrupted by some ``error string'' $E\in\pmset{m}$, which negates some of the bits of $C(x)$
(below, $C(x)\circ E$ denotes the entry-wise product of the two $m$-bit vectors $C(x)$ and $E$).
The oracle ``queries index $j\in[m]$'' if it reads the $j$'th bit of~$y$.
We use $\mathcal{A}^{y}(i)$ to denote the $\pm1$-valued random variable that is the algorithm's output.

\begin{definition}[Locally decodable code]\label{ldc:def}
A function $C: \pmset{n}\to\pmset{m}$ is a $(q,\delta,\eps)$-\emph{locally decodable code} if there exists a probabilistic oracle algorithm $\mathcal{A}$ such that
\begin{enumerate}
\item For every $x\in\pmset{n}$, every $i\in[n]$, and every $E\in\pmset{m}$ with at most $\delta m$ $-1$'s, we have $\Pr[\mathcal{A}^{C(x)\circ E}(i)=x_i]\geq 1/2+\eps$, where the probability is taken over the internal coin tosses of~$\mathcal{A}$.
\item $\mathcal{A}$ queries at most $q$ indices of $y$. 
Queries are made non-adaptively, meaning that the indices to be queried are all selected before the querying starts.
\end{enumerate}
An algorithm $\mathcal{A}$ satisfying the above is called a $(q,\delta,\eps)$-\emph{local decoder} for~$C$.
\end{definition}

Since any $\delta m$ indices can be corrupted, a local decoder must query the indices fairly uniformly. Otherwise, an adversary could choose to corrupt the most queried part of the code and ruin the decoder's success probability. Motivated by this property, Katz and Trevisan~\cite{katz&trevisan:ldc} defined a variation of a locally decodable code called a \emph{smooth code}, defined only for uncorrupted codewords.

\begin{definition}[Smooth code]
A function $C:\pmset{n}\to\pmset{m}$ is a $(q,c,\eps)$-\emph{smooth code} if there exists a probabilistic oracle algorithm $\mathcal{A}$ such that:
\begin{enumerate}
\item For every $x\in\pmset{n}$ and $i\in[n]$, we have $\Pr[\mathcal{A}^{C(x)}(i)=x_i]\geq 1/2 + \eps$.
\item For every $i\in[n]$ and $j\in[m]$, we have $\Pr[\mathcal{A}^{(\cdot)}(i) \text{ queries index } j]\leq c/m$.
\item $\mathcal{A}$ queries at most $q$ indices (non-adaptively).
\end{enumerate}
An algorithm satisfying the above is called a $(q,c,\eps)$-smooth decoder for $C$.
\end{definition}

Katz and Trevisan showed that LDCs and smooth codes are essentially equivalent, in the sense that a decoder for one can be transformed into a decoder for the other. We will prove the same for quantum codes in Section~\ref{qcodes:section}, using essentially their proof.

\subsection{Randomized codes}\label{randcodes:section}

Here we define our first generalization, incorporating randomness into the definition of the code.
A \emph{randomized locally decodable code} (randomized LDC) maps $\pmset{n}$ to \emph{random variables} over $\pmset{m}$ (rather than fixed codewords), such that any $x_i$ can be decoded well using a constant number of queries, even if up to $\delta m$ indices are corrupted. The formal definition is as follows.

\begin{definition}[Randomized locally decodable code]\label{rldcode}
A function $R:\pmset{n}\to\probdist{\pmset{m}}$ is a \emph{$(q,c,\eps)$-randomized locally decodable code} if there exists a probabilistic oracle algorithm $\mathcal{A}$ such that:
\begin{enumerate}
\item For every $x\in\pmset{n}$, every $i\in[n]$, and every $E\in\pmset{m}$ with at most $\delta m$ $-1$'s, we have $\Pr[\mathcal{A}^{R(x)\circ E}(i)=x_i]]\geq 1/2+\eps$, where the probability is taken over the internal coin tosses of~$\mathcal{A}$ as well as the distribution $R(x)$.
\item $\mathcal{A}$ queries at most $q$ indices (non-adaptively).
\end{enumerate}
An algorithm $\mathcal{A}$ satisfying the above is called a $(q,\delta,\eps)$-\emph{randomized local decoder} for~$R$.
\end{definition}

Similarly, we define a \emph{randomized smooth code}:

\begin{definition}[Randomized smooth code]\label{rscode}
A function $R:\pmset{n}\to\probdist{\pmset{m}}$ is a \emph{$(q,c,\eps)$-randomized smooth code} if there exists a probabilistic oracle algorithm $\mathcal{A}$ such that:
\begin{enumerate}
\item For every $x\in\pmset{n}$ and every $i\in[n]$, we have $\Pr[\mathcal{A}^{R(x)}(i) = x_i]\geq 1/2 + \eps$.
\item For every $i\in[n]$, and every $j\in[m]$, we have $\Pr[\mathcal{A}^{(\cdot)}(i) \text{ queries index } j]\leq c/m$.
\item $\mathcal{A}$ queries at most $q$ indices (non-adaptively).
\end{enumerate}
An algorithm $\mathcal{A}$ satisfying the above is called a $(q,c,\eps)$-randomized smooth decoder for~$R$.
\end{definition}

It will be convenient to also have a version of these codes that are only required to work well on average, instead of for all $x$:

\begin{definition}[$\mu$-average codes]\label{avrscode}
Let $\mu$ be a distribution on $\pmset{n}$.
A function $C:\pmset{n}\to\pmset{m}$ is a \emph{$\mu$-average $(q,\delta,\eps)$-locally decodable code} if Definition~\ref{ldc:def} holds with the first clause replaced by:
\begin{enumerate}
\item For every $i\in[n]$ and $E\in\pmset{m}$ with at most $\delta m$ $-1$'s,
$\Pr_{x\sim\mu}[\mathcal{A}^{C(x)\circ E}(i)=x_i]\geq \frac{1}{2}+\eps$.
\end{enumerate}
Analogously, we define $\mu$-average versions of smooth codes, randomized LDCs, and randomized smooth codes.
For these codes, we assume without loss of generality that for each $i$ and queried set $r\subseteq[m]$, the decoder $\mathcal{A}$ always uses
the same function $f_{i,r}:\pmset{q}\rightarrow\{\pm 1\}$ to determine its output.
\end{definition}

A $\mu$-average randomized smooth code can actually be ``derandomized'' to a $\mu$-average smooth code on a smaller number of bits:

\begin{lemma}\label{rsc=sc:lem}
Let $R:\pmset{n}\to\probdist{\pmset{m}}$ be a $\mu$-average $(q,c,\eps)$-randomized smooth code. Then there exists a $\mu$-average $(q,c,\eps/2)$-smooth code $C:\pmset{n}\to\pmset{m}$ for at least $\eps n$ of the indices $i$ (that is, a smooth code with $\mu$-success probability at least $1/2 + \eps/2$ for at least $\eps n$ of the $n$ indices).
\end{lemma}

\begin{proof}
As a first step we will view $R$ as a function to strings: 
there exists a random variable $Z$ (over some possibly infinite set $\mathcal{Z}$) 
and a function $R:\pmset{n}\times Z\to\pmset{m}$ such that
for every $x\in\pmset{n}$, the random variables $R(x,Z)$ and $R(x)$ are the same. 
A decoder $\mathcal{A}$ for $R$ also works for $R(\cdot,Z)$, so 
we have bias $\Exp_{x\sim\mu,z\sim Z}[\mathcal{A}^{R(x,z)}(i)\cdot x_i]\geq 2\eps$ for every $i\in[n]$.
For every $i\in[n]$ and $z\in\mathcal{Z}$, define variables $X_{i,z}\in\01$, with 
$$
X_{i,z}=1\Longleftrightarrow\Exp_{x\sim\mu}[\mathcal{A}^{R(x,z)}(i) \cdot x_i]\geq\eps,
$$ 
and $X_z :=\sum_{i=1}^n X_{i,z}$. Using the definition of a $\mu$-average randomized smooth code, we have
\beqrn
2\eps n &\leq& \sum_{i=1}^n\Exp_{x\sim\mu,z\sim Z}[\mathcal{A}^{R(x,z)}(i)\cdot x_i]\\
&=& \Exp_{z\sim Z}\left[\sum_{i=1}^n\Exp_{x\sim\mu}[\mathcal{A}^{R(x,z)}(i)\cdot x_i]\right]\\
&<& \Exp_{z\sim Z}[X_z + (n-X_z)\eps]\\
&=& \eps n + (1-\eps)\Exp_{z\sim Z}[X_z].
\eeqrn
Hence $\Exp_{z\sim Z}[X_z]\geq\eps n$. 
Thus there exists a $z\in\mathcal{Z}$ such that for at least $\eps n$ of the $n$ indices $i$, we have
$\Exp_{x\sim\mu}[\mathcal{A}^{R(x,z)}(i) \cdot x_i]\geq\eps$,
equivalently, $\Exp_{x\sim \mu}[\Pr[\mathcal{A}^{R(x,z)}(i) = x_i]]\geq 1/2 + \eps/2$. 
Defining the code $C(\cdot):=R(\cdot,z)$ gives the lemma.
\end{proof}

\subsection{Quantum codes}\label{qcodes:section}

Our second level of generalization brings quantum mechanics into the picture: now our code maps classical $n$-bit strings to $m$-qubit quantum states. Decoding of these codes requires algorithms that use both quantum measurements and properties of classical probabilistic oracle algorithms. Below, with ``quantum oracle algorithm'' we mean an algorithm $\mathcal{A}$ with oracle access to an $m$-qubit state~$\rho$,
which is written as a superscript.  This $\rho$ could be a corrupted version of an $m$-qubit ``codeword'' $Q(x)$, 
obtained by applying some super-operator $\mathcal{E}$ to $Q(x)$. 
This $\mathcal{E}$ should only affect a $\delta$-fraction of the $m$ qubits.
This way of modelling the error generalizes the classical case: 
a classical error pattern $E\in\pmset{m}$ corresponds to a super-operator $\mathcal{E}$ that applies an $X$ to the qubits at 
positions where $E$ has a $-1$, and $I$ to the positions where $E$ has a $+1$.
On input $i\in[n]$, the algorithm probabilistically selects a set $r\subseteq[m]$ of at most $q$ indices of qubits
of this state, and applies a two-outcome measurement to the selected qubits with operators $A^+_{i,r}$ and $A^-_{i,r}$.
As before, we will use ``$\mathcal{A}^{\rho}(i)$'' to denote the $\pm 1$-valued random variable that is the output.
We say that ``$\mathcal{A}$ queries $r$'', and ``$\mathcal{A}$ queries index $j$'' if $j$ is in $r$.
Note that such algorithms are non-adaptive by definition: they first select the qubits in $r$, and then apply one measurement to those qubits.

We now define a \emph{locally decodable quantum code} (LDQC) as follows:

\begin{definition} [Locally decodable quantum code] 
A function $Q:\{\pm 1\}^n\to\densop{2^m}$ is a $(q,\delta,\eps)$-\emph{locally decodable quantum code} if there exists a quantum oracle algorithm $\mathcal{A}$ such that: 
\begin{enumerate}
\item For every $x\in\pmset{n}$, every $i\in[n]$, and every super-operator $\mathcal{E}$ that acts non-trivially on at most $\delta m$ qubits, we have $\Pr[\mathcal{A}^{\mathcal{E}(Q(x))}(i)=x_i]\geq 1/2 + \eps$, where the probability is taken over the coin tosses and measurements in $\mathcal{A}$.
\item $\mathcal{A}$ queries at most $q$ indices (non-adaptively).
\end{enumerate}
An algorithm $\mathcal{A}$ satisfying the above requirements is called a $(q,\delta,\eps)$-local quantum decoder for $Q$.
\end{definition}

LDQCs generalize randomized LDCs, because probability distributions are just diagonal density operators.
Similarly, we can establish a smoothness property also for quantum codes:

\begin{definition}[Smooth quantum code]
A function $Q:\{\pm 1\}^n\to\densop{2^m}$ is a $(q,c,\eps)$-\emph{smooth quantum code} if there exists a quantum oracle algorithm $\mathcal{A}$ such that:
\begin{enumerate}
\item For every $x\in\{\pm 1\}^n$ and every $i\in[n]$, we have $\Pr[\mathcal{A}^{Q(x)}(i)=x_i]\geq 1/2 + \eps$.
\item For every $i\in[n]$ and every $j\in[m]$, we have $\Pr[\mathcal{A}^{(\cdot)}(i) \text{ queries index j }] \leq c/m$.
\item $\mathcal{A}$ queries at most $q$ indices (non-adaptively).
\end{enumerate}
An algorithm $\mathcal{A}$ satisfying the above is called a $(q,c,\eps)$-smooth quantum decoder for $Q$. 
\end{definition}

As Katz and Trevisan~\cite{katz&trevisan:ldc} did for classical LDCs, we can establish a strong connection between LDQCs and smooth quantum codes. Either one can be used as the other, as the next theorems show. Analogues of these theorems also hold between randomized LDCs and randomized smooth codes, and between the $\mu$-average versions of these codes.

\begin{theorem}\label{sqctoldqc}
Let $Q:\pmset{n}\to\densop{2^m}$ be a $(q,c,\eps)$-smooth quantum code. Then, as long as  $\delta \leq \eps/c$, we have that $Q$ is also a $(q,\delta,\eps - \delta c)$-locally decodable quantum code.
\end{theorem}

\begin{proof}
Let $\mathcal{A}$ be a $(q,c,\eps)$-smooth quantum decoder for $Q$. 
Suppose we run it on $\mathcal{E}(Q(x))$ with at most $\delta m$ corrupted qubits.
The probability that $\mathcal{A}$ queries a specific qubit is at most $c/m$.
Then by the union bound, the probability that $\mathcal{A}$ queries any of the corrupted qubits is at most $\delta m c/m=\delta c$. 
Hence $\mathcal{A}$ itself is also a $(q,\delta,\eps-\delta c)$-local quantum decoder for $Q$.
\end{proof}

\begin{theorem}\label{ldqctosqc}
Let $Q:\pmset{n}\to\densop{2^m}$ be a $(q,\delta,\eps)$-locally decodable quantum code. Then $Q$ is also a $(q,q/\delta,\eps)$-smooth quantum code.
\end{theorem}

\begin{proof}
Let $\mathcal{A}$ be a $(q,\delta,\eps)$-local quantum decoder for $Q$. For each $i\in[n]$, let $p_i(j)$ be the probability that on input~$i$, $\mathcal{A}$ queries qubit $j$. Let $H_i=\set{j\mid p_i(j)>q/(\delta m)}$. Then $|H_i|\leq\delta m$, because $\mathcal{A}$ queries no more than $q$ indices.
Let $\mathcal{B}$ be the quantum decoder that simulates $\mathcal{A}$, except that on input $i$ it does not query qubits in $H_i$, but instead acts as if those qubits are in a completely mixed state.
Then $\mathcal{B}$ does not measure any qubit $j$ with probability greater than $q/(\delta m)$.
Also, $\mathcal{B}$'s behavior on input $i$ and $Q(x)$ is the same as $\mathcal{A}$'s behavior on input $i$ and $\mathcal{E}(Q(x))$ that is obtained by replacing all qubits in $H_i$ by completely mixed states. Since $\mathcal{E}$ acts non-trivially on at most $|H_i|\leq\delta m$ qubits, we have $\Pr[\mathcal{B}^{Q(x)}(i)=x_i] = \Pr[\mathcal{A}^{\mathcal{E}(Q(x))}(i)=x_i]\geq 1/2 + \eps$.
\end{proof}

\subsection{A weak lower bound from random access codes}\label{ssecqrac}
We can immediately establish a weak lower bound on the length of LDQCs and smooth quantum codes by considering a \emph{quantum random access code} (QRAC)~\cite{antv:densej}, which generalizes both.

\begin{definition}[Quantum random access code]
A function $Q: \pmset{n}\to \densop{2^m}$ is an $(n,m,\eps)$-\emph{quantum random access code} if there exists a quantum oracle algorithm $\mathcal{A}$ such that for every $x\in\pmset{n}$ and $i\in [n]$, $\Pr[\mathcal{A}^{Q(x)}(i)=x_i] \geq 1/2 + \eps$. An algorithm $\mathcal{A}$ that satisfies this is called a quantum random access decoder for $Q$.
\end{definition}

LDQCs and smooth quantum codes are QRACs with some additional properties, such as constraints on the way the qubits of the codeword are accessed. Hence the following well-known lower bound on the length of QRACs also holds for them.

\begin{theorem}[ANTV~\cite{antv:densej,nayak:qfa}]\label{qracbound}
Every $(n,m,\eps)$-QRAC satisfies $m\geq (1-H(1/2 + \eps))n$.
\end{theorem}

\section{Pauli decoding from disjoint subsets}

In this section we consider a $(q,c,\eps)$-smooth quantum code $Q$.  Fix a distribution $\mu$ on $\pmset{n}$. We will show that there exists a sequence $S^*\in{\cal P}_m$ such that if the $m$ qubits of $Q(x)$ are measured by the $m$ Pauli measurements in $S^*$, then each $x_i$ can be retrieved by querying only $q$ bits of the $m$-bit measurement outcome $S^*(Q(x))$, in a very structured way. Specifically, we prove:

\begin{theorem}\label{specialS}
Let $Q:\pmset{n}\to\densop{2^m}$ be a $(q,c,\eps)$-smooth quantum code and $\mu$ be a distribution on $\pmset{n}$. Then there exists a sequence $S^*\in{\mathcal P}_m$, and for every $i\in[n]$ a set $M_i$ of at least $\eps m/(qc)$ disjoint sets $r\subseteq[m]$ (each of size at most $q$) with associated signs $a_{i,r}\in\pmset{}$, such that
\beqn
\Exp_{x\sim\mu}\left[\frac{1}{|M_i|}\sum_{r\in M_i}\Pr\big[a_{i,r}\prod_{j\in r}S^*_j(Q(x))=x_i\big]\right]\geq \half + \frac{\eps}{4^{q+1}}.
\eeqn
\end{theorem}

The proof consists of two parts. We start by constructing the sets $M_i$ and then we show that decoding $Q$ can be done by using only Pauli measurements. Putting these two observations together enables us to prove Theorem \ref{specialS}.

As an aside, the fact that this theorem works for every distribution $\mu$ allows us to turn smooth codes into schemes for private information retrieval (PIR) that work for every $x\in\pmset{n}$ instead of only on average. We explain this in Appendix \ref{rsctopir}.

\subsection{Decoding from disjoint subsets}

First we construct the large sets $M_i$ of disjoint $q$-sets that enable reasonably good prediction of $x_i$.

\begin{theorem}[modified from Lemma~4 in \cite{katz&trevisan:ldc}]\label{matchinglem}
Let $Q:\pmset{n}\to\densop{2^m}$ be a $(q,c,\eps)$-smooth quantum code with decoder $\mathcal{A}$, and $\mu$ a distribution on $\pmset{n}$. 
Then for every $i\in[n]$ there exists a set $M_i$ of at least $\eps m/(qc)$ disjoint sets $r\subseteq[m]$ (each of size at most $q$) satisfying
$$
\Pr_{x\sim\mu}[\mathcal{A}^{Q(x)}(i)=x_i\mid \mathcal{A}^{(\cdot)}(i) \text{ queries } r]\geq\frac{1}{2} + \frac{\eps}{2}.
$$
\end{theorem}

\begin{proof}
Call a set $r\subseteq[m]$ ``good for $i$'' if it satisfies the inequality stated in the theorem.
Define for every $i\in[n]$ a hypergraph $H_i=(V,E_i)$ with vertex-set $V=[m]$ and a set of hyperedges $E_i := \{e\mid e \text{ is good for } i\}$. Say that a smooth quantum decoder $\mathcal{A}$ for $Q$ ``queries $E_i$'' if $\mathcal{A}$ queries an $e\in E_i$. Let $p(e) := \Pr[\mathcal{A}^{(\cdot)}(i) \text{ queries } e]$. Then the probability that this decoder queries $E_i$ is $p(E_i):=\sum_{e\in E_i}p(e)$. For all $e\not\in E_i$ we have
$$
\Pr[\mathcal{A}^{Q(x)}(i) = x_i\mid\mathcal{A}^{(\cdot)}(i) \text{ queries } e] < \frac{1}{2} + \frac{\eps}{2}.
$$
But since for every $x$ and $i$, $\mathcal{A}$ decodes bit $x_i$ with probability at least $1/2+\eps$, we have
$$
\frac{1}{2} + \eps\leq \Pr[\mathcal{A}^{Q(x)}(i)=x_i] < p(E_i) + (1-p(E_i))(\frac{1}{2} + \frac{\eps}{2}) =\frac{1}{2} + \frac{\eps}{2} + p(E_i)(\frac{1}{2} - \frac{\eps}{2}).
$$
Hence $p(E_i) > \eps/(1-\eps) \geq \eps$. 
Since $Q$ is a \emph{smooth} quantum code, we know that the probability that $\mathcal{A}$ queries an index $j$ is $\sum_{e\in E_i|j\in e}p(e) = \Pr[\mathcal{A}^{(\cdot)}(i) \text{ queries } j] \leq c/m$. 

Let $M_i$ be a maximal set of disjoint hyperedges in $H_i$, and define the vertex set $T = \cup_{e\in M_i}e$. Note that $T$ has at most $q|M_i|$ elements and that it intersects each $e\in E_i$ (since otherwise $|M_i|$ would not be maximal). We can now lower bound the size of $M_i$ as follows:
$$
\eps < p(E_i) = \sum_{e\in E_i}p(e) \stackrel{(*)}{\leq} \sum_{j\in T}\sum_{e\in E_i|j\in e}\!p(e) \leq \frac{c|T|}{m}\leq \frac{cq|M_i|}{m},
$$
where $(*)$ holds because each $e\in E_i$ is counted exactly once on the left hand side, and at least once on the right-hand side (since $T$ intersects each $e\in E_i$). Hence $|M_i|> \eps m/(qc)$.
\end{proof}

\subsection{Pauli decoding}\label{paulidecoding}

In the second part of the proof of Theorem~\ref{specialS}, we find the appropriate Pauli measurements. Recall that to decode $x_i$, a smooth quantum decoder first selects a set $r\subseteq[m]$ of at most $q$ indices, and then applies some measurement with operators $A_{i,r}^+,A_{i,r}^-$ to determine its output.  Let $A_{i,r}=A_{i,r}^+-A_{i,r}^-$. Strictly speaking these operators act only on the qubits indexed by $r$, but we can view them as acting on the $m$-qubit state $Q(x)$ by tensoring them with $m-|r|$ identities. The difference between the probabilities of obtaining outcomes $+1$ and $-1$ is $\Tr(A_{i,r}\cdot Q(x))$. For every $i\in[n]$ and $r\in M_i$ we define the following bias:
$$
B(i,r):=\Exp_{x\sim\mu}[\Tr(A_{i,r}\cdot Q(x))\cdot x_i].
$$
This measures how well the measurement outcome is correlated with $x_i$ (with $x$ weighted according to $\mu$).
From Theorem~\ref{matchinglem} we have $B(i,r)\geq \eps$ for every $i\in[n]$ and every $r\in M_i$.

Since ${\cal P}_q$ is a basis for all $2^q\times 2^q$ complex matrices we can write
$$
A_{i,r}=\sum_{S\in{\cal P}_q}\widehat{A_{i,r}}(S) S,
$$
with $\widehat{A_{i,r}}(S):=\langle A_{i,r},S\rangle=\frac{1}{2^q}\Tr(A_{i,r}\cdot S)\in[-1,1]$. We now have:
\begin{equation}\label{eqsumofbiases}
\eps\leq B(i,r)=\sum_{S\in{\cal P}_q}\widehat{A_{i,r}}(S)\Exp_{x\sim\mu}[\Tr(S \cdot Q(x))\cdot x_i]
\leq \sum_{S\in{\cal P}_q}\left|\Exp_{x\sim\mu}[\Tr(S \cdot Q(x))\cdot x_i]\right|.
\end{equation}
Suppose we measure the $r$-qubits of $Q(x)$ with some $S\in{\cal P}_q$ and get outcome $b\in\{\pm 1\}$.
The quantity $\Exp_{x\sim\mu}[\Tr(S \cdot Q(x))\cdot x_i]$ is the difference between $\Pr_{x\sim\mu}[b=x_i]$ and $\Pr_{x\sim\mu}[b\neq x_i]$. If we output $b$ if this difference is nonnegative, and $-b$ otherwise, then we would predict $x_i$ with bias
$$
B'(i,S,r):=\left|\Exp_{x\sim\mu}[\Tr(S \cdot Q(x))\cdot x_i]\right|.
$$
From Equation~(\ref{eqsumofbiases}) we know that this bias is at least $\eps/4^q$ for at least one ``good'' $S\in{\cal P}_q$.
Hence, with some loss in success probability, we can decode $Q$ by only using Pauli measurements.
We now use a probabilistic argument to prove that a good sequence $S^*$ of Pauli measurements exists,
which is simultaneously good, for every $i\in[n]$, for most of the elements $r\in M_i$.

\begin{proof}[ (of Theorem~\ref{specialS})]
Suppose we let $\mathbf{S}\in{\cal P}_q$ be a random variable uniformly distributed over ${\cal P}_q$, and we use it to predict $x_i$ as above.
Then $B'(i,\mathbf{S},r)$ is a random variable in the interval $[0,1]$, with expectation
$$
\Exp_{S\in{\cal P}_q}[B'(i,S,r)]=\frac{1}{4^q}\sum_{S\in{\cal P}_q}\left|\Exp_{x\sim\mu}[\Tr(S\cdot Q(x))\cdot x_i]\right|\geq\frac{\eps}{4^q}.
$$
Now we consider $m$-qubit Pauli measurements and replace all elements not in $r$ with $I$'s:
for $S\in{\cal P}_m$ and $r\subseteq[m]$, let $S_{(r)}$ denote $S$ with all its $m-|r|$ elements outside of $r$ replaced by $I$.
If we let $\mathbf{S}$ be uniform over ${\cal P}_m$, we get biases $B'(i,\mathbf{S}_{(r)},r)$ for each $r\in M_i$,
each in $[0,1]$ and with expectation at least $\eps/4^q$ (over the choice of $S_{(r)}$). But note that the
random variables $B'(i,\mathbf{S}_{(r)},r)$ are independent from each other for different $r\in M_i$, since the
elements of $M_i$ are disjoint. Hence the average bias over all $r\in M_i$,
$$
B'(\mathbf{S},i):=\frac{1}{|M_i|}\sum_{r\in M_i}B'(i,\mathbf{S}_{(r)},r),
$$
is the average of $|M_i|$ independent random variables,
each in $[0,1]$ and with expectation at least $\eps/4^q$.
By a Chernoff bound\footnote{See Equation~(7) in~\cite{hagerup:tour}. A small modification of their proof shows that this bound not only holds for independent 0/1-variables, but also for independent variables in the interval $[0,1]$.} the probability that $B'(\mathbf{S},i)$ is much smaller than its expectation, is small:
$$
\Pr_{S\in{\cal P}_m}\left[B'(S,i)<\frac{1}{2}\frac{\eps}{4^q}\right]\leq \Pr_{S\in{\cal P}_m}\left[B'(S,i)<\frac{1}{2}\Exp[B'(S,i)]\right]\leq \exp\left(-\frac{|M_i|\eps}{8\cdot 4^q}\right).
$$
By Theorems~\ref{qracbound} and~\ref{matchinglem} we may assume $|M_i|>4^q\log(n)/\eps$.
It follows that the above probability is less than $1/n$.
Since this is true for every index $i\in[n]$, the union bound gives
$$
\Pr_{S\in{\cal P}_m}\left[\exists i~s.t.~B'(S,i)<\frac{1}{2}\frac{\eps}{4^q}\right]
\leq\sum_{i=1}^n \Pr_{S\in{\cal P}_m}\left[B'(S,i)<\frac{1}{2}\frac{\eps}{4^q}\right]<1.
$$
We can thus conclude that there exists an $S^*\in{\cal P}_m$ such that for every $i\in[n]$ we have
$$
\frac{1}{|M_i|}\sum_{r\in M_i}B'(i,S^*_{(r)},r)\geq\frac{1}{2}\frac{\eps}{4^q}.
$$
This implies the statement of the theorem.
\end{proof}

\section{Classical codes from quantum codes}

Theorem \ref{specialS} implies that if we measure all $m$ indices of a smooth quantum quantum code $Q$ with the elements of $S^*$, then we get distributions on $\pmset{m}$ that can be massaged to ``codewords'' $R(x)$ of a randomized smooth code:

\begin{theorem}\label{sqc=rsc:thm}
Let $Q:\pmset{n}\to\densop{2^m}$ be a $(q,c,\eps)$-smooth quantum code. Then for every input distribution $\mu$ on $\{\pm1\}^n$, 
there exists a $\mu$-average $(q,qc/\eps,\eps/4^{q+1})$-randomized smooth code $R:\pmset{n}\to\probdist{\pmset{m}}$.
\end{theorem}

\begin{proof}
We use Theorem~\ref{specialS}.
Let $R(x)$ be the distribution on $\pmset{m}$ obtained by measuring $Q(x)$ with $S^*$.
We define a decoder $\mathcal{A}$ for $R$ as follows: on input $i\in[m]$ and oracle $y\in\pmset{m}$,
pick a set $r$ from the set $M_i$ uniformly at random, and return $a_{i,r}\prod_{j\in r}y_j$.
It is straightforward to check that $\mathcal{A}$ is a $\mu$-average $(q,qc/\eps,\eps/4^{q+1})$ decoder for $R$;
in particular, since $\mathcal{A}$ picks $r$ uniformly from a set of at least $\eps m/(qc)$ disjoint sets,
each index $j\in[m]$ has probability at most $qc/(\eps m)$ of being queried.
\end{proof}

Combining Lemma~\ref{rsc=sc:lem} and Theorem~\ref{sqc=rsc:thm}, we immediately get the following ``derandomization'':

\begin{corollary}\label{sqc=sc:cor}
Let $Q:\pmset{n}\to\densop{2^m}$ be a $(q,c,\eps)$-smooth quantum code. Then for every distribution $\mu$ on $\pmset{n}$,
there exists a $C:\pmset{n}\to\pmset{m}$ which is a $\mu$-average $(q,qc/\eps,\eps/(2\cdot 4^{q+1}))$-smooth code for at least $\eps n/4^{q+1}$ of the $n$ indices.
\end{corollary}

Following the path through Theorems~\ref{ldqctosqc}, \ref{sqc=rsc:thm}, and the $\mu$-average version of Theorem~\ref{sqctoldqc}, 
we can turn an LDQC into a $\mu$-average randomized LDC:

\begin{corollary}\label{ldqc=rldc:cor}
Let $Q:\pmset{n}\to\densop{2^m}$ be a $(q,\delta,\eps)$-locally decodable quantum code.
Then, as long as $\delta'\leq \delta\eps^2/(q^2 4^{q+1})$, for every distribution $\mu$ over $\pmset{n}$, 
there exists an $R:\pmset{n}\to\probdist{\pmset{m}}$ which is a $\mu$-average $(q,\delta',\eps/4^{q+1} - \delta' q^2/(\delta\eps))$-randomized locally decodable code.
\end{corollary}

Going through Theorem~\ref{ldqctosqc}, Corollary~\ref{sqc=sc:cor}, and the $\mu$-average version of Theorem~\ref{sqctoldqc} instead, 
we can also turn an LDQC into a $\mu$-average LDC:

\begin{corollary}\label{ldqc=ldc:cor}
Let $Q:\pmset{n}\to\densop{2^m}$ be a $(q,\delta,\eps)$-locally decodable quantum code.
Then, as long as $\delta'\leq\delta\eps^2/(2 q^2 4^{q+1})$, for every distribution $\mu$ over $\pmset{n}$,
there exists a $C:\pmset{n}\to\pmset{m}$ which is a $\mu$-average 
$(q,\delta',\eps/(2\cdot 4^{q+1}) - \delta' q^2/(\delta\eps))$-locally decodable code for at least $\eps n/4^{q+1}$ of the $n$ indices.
\end{corollary}


\section{Conclusion and open problems}

We defined quantum generalizations of $q$-query locally decodable codes in which $q$ queries correspond to a measurement on $q$ qubits of the $m$-qubit codeword. By a reduction to (classical) randomized smooth codes through a special sequence of Pauli measurements on an LDQC, we showed that the use of quantum systems for this type of encoding can not provide much advantage in terms of length, at least for small $q$. An obvious open problem is reducing the gap between upper and lower bound on the length $m$ of LDCs for fixed small number of queries $q$. Our results show that an upper bound for LDQCs would carry over to ($\mu$-average) LDCs.  This might perhaps be a way to improve the best known classical upper bounds on $m$.

\subsection*{Acknowledgments}
We thank Harry Buhrman, Peter H{\o}yer, Oded Regev, Falk Unger and Stephanie Wehner for useful discussions. JB is especially indebted to Peter H\o yer for suggesting this problem while being his guest at the University of Calgary, Institute for Quantum Information Science.

\bibliographystyle{alpha}

\appendix

\section{PIR schemes from LDCs}\label{rsctopir}

Katz and Trevisan~\cite{katz&trevisan:ldc} showed that LDCs are closely related to so-called private information retrieval (PIR) schemes,
first introduced by Chor et al.~\cite{cgks:pir}. In a PIR scheme,
$q$ non-communicating ``servers'' each hold a copy of the same database $x\in\pmset{n}$.
A ``user'' interacts (usually in only one round of communication) with these servers to retrieve the $i$'th bit
$x_i$ while preserving privacy:
individually, the servers should get no information whatsoever about which index $i$ the user is interested in.
The resource to be minimized is the amount of communication between user and servers.

Katz and Trevisan observed that a smooth code implies a PIR scheme where the user has good recovery probability \emph{on average}.
Specifically, Theorem~\ref{matchinglem} gives a $(q,q,\eps^2/(2c))$-smooth decoder as follows.
We complete the set $M_i$ to a set $M'_i$ of exactly $m/q$ disjoint $q$-tuples (assume for simplicity that $q$ divides $m$).
Now the decoder uniformly picks an $r\in M_i'$
and queries those $q$ indices. If $r$ contains an element of $M_i$ (which happens with probability at least $\eps/c$)
then the decoder proceeds as before, predicting $x_i$ with probability at least $1/2+\eps/2$ (under $\mu$); otherwise the decoder outputs a fair coin flip.
Note that for each $i\in[n]$, the overall success probability (under $\mu$) is at least $(\eps/c)(1/2+\eps/2)+(1-\eps/c)/2=1/2+\eps^2/(2c)$.
Also, each index $j\in[m]$ is queried with probability exactly $q/m$.
Thus we have a $\mu$-average $(q,q,\eps^2/(2c))$-smooth code.

This in turn gives a PIR scheme with good success probability under $\mu$: the user just sends one query to each of the servers,
the servers return the requested bit of the code, and the user gives the same output as the code's decoder.  Since each query individually is uniformly distributed, no information about $i$ will be leaked to individual servers.\footnote{The same argument works to derive PIR schemes from codes over a non-binary alphabet, where the servers' answers are more than one bit. Conversely, one can get a smooth code from a one-round PIR scheme by, roughly speaking, concatenating all answers of the $q$ servers to all possible messages that the user can send them.}

However, we can actually show that there exists a $(q,q,\eps^2/(2c))$-smooth decoder that can decode any bit $x_i$ \emph{for every} 
$x\in\pmset{n}$, hence giving true PIR schemes that work for every database instead of only on $\mu$-average.

\begin{theorem}
Let $C:\pmset{n}\to\pmset{m}$ be a $(q,c,\eps)$-smooth code. Then there exists a $(q,q,\eps^2/(2c)$)-smooth decoder for $C$.
\end{theorem}

\begin{proof}

Fix an $i\in[n]$.
We will show that there exists a decoder $\mathcal{B}(i)$ such that for all $x\in\pmset{n}$
\beqn
\Pr[\mathcal{B}^{C(x)}(i)=x_i]\geq \frac{1}{2} + \frac{\eps^2}{2c}.
\eeqn
Consider all possible pairs $(M,F)$, where $M$ is a set of at least $\eps m/(qc)$ disjoint sets $r\subseteq[m]$, each of size at most $q$, 
and $F$ contains one Boolean function $f_r$ for each $r\in M$.
Define a decoder $\mathcal{A}_{(M,F)}$ that decodes according to this pair,
i.e., it queries a uniformly random $r\in M$ and applies $f_r\in F$ to the results.
Define a matrix $P$, with rows indexed by all $x$ and columns by pairs $(M,F)$:
$$
P_{x,(M,F)}=\Pr[\mathcal{A}_{(M,F)}^{C(x)}=x_i].
$$
Theorem \ref{matchinglem} says that for every distribution $\mu$ over $\pmset{n}$, there exists a column of $P$
(i.e., an $(M,F)$ pair) with $\mu$-average at least $1/2+\eps/2$.

For each $x$ define $2^n$-dimensional 0/1-vector $e_x$ with a 1 only at position $x$, and similarly define 0/1-vector $u_{(M,F)}$.
For probability distributions $\mu$ (on the set of all $x$) and $\nu$ (on the set of all pairs $(M,F)$), 
define vectors $e_{\mu} = \sum_{x\in\pmset{n}}\mu(x)e_x$ and $u_{\nu} = \sum_{(M,F)}\nu((M,F))u_{(M,F)}$.
Then Yao's principle (i.e., the minimax theorem as used in~\cite{yao:unified}) gives us
\beqrn
\frac{1}{2} + \frac{\eps}{2} &\leq& \min_{\mu}\max_{(M,F)} e_{\mu}^T P u_{(M,F)} =\max_{\nu}\min_x e^T_x P u_{\nu}.
\eeqrn
Let $\nu$ be a distribution that maximizes the right-hand side.
Let $\mathcal{B}(i)$ select a pair $(M,F)$ according to distribution $\nu$, complete $M$ to some $M'$, and use a uniformly chosen element of $M'_i$ to predict~$x_i$, as explained before the theorem.  Then $\mathcal{B}(i)$ queries every index with probability exactly $q/m$, and satisfies
$$
\min_{x\in\pmset{n}}\Pr[\mathcal{B}^{C(x)}(i)=x_i]\geq \frac{1}{2}+\frac{\eps^2}{2c}.
$$
Hence the algorithms $\mathcal{B}(1),\ldots,\mathcal{B}(n)$ form a $(q,q,\eps^2/(2c))$-smooth decoder for $C$.
\end{proof}

\end{document}